\documentclass{amsart}

\usepackage[english]{babel}
\usepackage[utf8x]{inputenc}
\usepackage[T1]{fontenc}

\usepackage{amsmath}
\usepackage{amsfonts}
\usepackage{mathrsfs}
\usepackage{bm}
\usepackage{graphicx}
\usepackage{spverbatim}
\usepackage{url}
\usepackage[colorinlistoftodos]{todonotes}
\usepackage[colorlinks=true, allcolors=blue]{hyperref}
\usepackage[ruled,vlined,linesnumbered]{algorithm2e}
\SetEndCharOfAlgoLine{}
\usepackage{tikz-cd}
\usetikzlibrary{arrows}

\newcommand{\thegraph}{\mathcal{G}}
\newcommand{\thegraphell}{\mathcal{E}} 
\newcommand{\thegraphjac}{\mathcal{J}}
\newcommand{\Z}{\mathbb{Z}}
\newcommand{\ZZ}{\Z}  
\newcommand{\Q}{\mathbb{Q}}
\newcommand{\QQ}{\Q}  

\newcommand{\F}{\mathbb{F}}
\newcommand{\FF}{\F}  
\newcommand{\FFbar}{\overline{\FF}}  

\newtheorem{theorem}{Theorem}
\newtheorem{proposition}{Proposition}
\newtheorem{conjecture}{Conjecture}
\newtheorem{problem}{Problem}

\theoremstyle{definition}
\newtheorem{definition}{Definition}
\newtheorem{remark}{Remark}
\newtheorem{example}{Example}

\newcommand{\splitatcommas}[1]{%
  \begingroup
  \begingroup\lccode`~=`, \lowercase{\endgroup
    \edef~{\mathchar\the\mathcode`, \penalty0 \noexpand\hspace{0pt plus 1em}}%
  }\mathcode`,="8000 #1%
  \endgroup
}

\begin{document}
\title[Hash functions from superspecial genus-$2$ curves]{Hash functions from superspecial genus-$2$ curves using Richelot isogenies}
\date{}
\author{Wouter Castryck}
\address{Section of Algebra, Department of Mathematics, KU Leuven}
\email{wouter.castryck@esat.kuleuven.be}
\author{Thomas Decru}
\address{imec-COSIC, Department of Electrical Engineering, KU Leuven}
\email{thomas.decru@esat.kuleuven.be}
\author{Benjamin Smith}
\address{Inria \textit{and} Laboratoire d’Informatique de l’\'Ecole polytechnique, Universit\'e Paris–Saclay, Palaiseau, France}
\email{smith@lix.polytechnique.fr}

\maketitle 

\begin{abstract}
Last year Takashima proposed a version of Charles, Goren and Lauter's hash function using Richelot isogenies, starting from a genus-$2$ curve that allows for all subsequent arithmetic to be performed over a quadratic finite field $\FF_{p^2}$. In a very recent paper Flynn and Ti point out that Takashima's hash function is insecure due to the existence of small isogeny cycles. We revisit the construction and show that it can be repaired by imposing a simple restriction, which moreover clarifies the security analysis. The runtime of the resulting hash function is dominated by the extraction of $3$ square roots for every block of $3$ bits of the message, as compared to one square root per bit in the elliptic curve case; however in our setting the extractions can be parallelized and are done in a finite field whose bit size is reduced by a factor $3$. Along the way we argue that the full supersingular isogeny graph is the wrong context in which to study higher-dimensional analogues of Charles, Goren and Lauter's hash function, and advocate the use of the superspecial subgraph, which is the natural framework in which to view Takashima's $\FF_{p^2}$-friendly starting curve.
\end{abstract}

\section{
    Introduction
}

After a cautious start with Couveignes' unpublished note~\cite{couveignes} from 1997 and Stolbunov's master thesis~\cite{stolbunov} from 2004, 
the area of isogeny-based cryptography took a more visible turn in 2006 when Charles, Goren and Lauter~\cite{CGL} showed how to construct collision-resistant hash functions from deterministic walks in isogeny graphs of supersingular elliptic curves over finite fields.
About five years later Jao and De Feo applied similar ideas to the design of a key exchange protocol~\cite{SIDH,jaodefeoplut} now known as SIDH, after which isogenies became a very active topic of cryptographic research, largely due to their promise of leading to quantum resistant hard problems. 
Some of the recent constructions include non-interactive key exchange~\cite{defeokieffersmith,csidh}, signatures~\cite{seasign,seasign2} and verifiable delay functions~\cite{VDF}.  
Last January it was announced that SIKE~\cite{SIKE}, which is an incarnation of SIDH, is one of the seventeen second-round contenders to become a NIST standard for post-quantum key establishment.\footnote{See \url{https://csrc.nist.gov/projects/post-quantum-cryptography/round-2-submissions}.} 

While almost all of the ongoing research in isogeny-based cryptography is devoted to elliptic curves, there is a general awareness that many proposals 
should generalize to principally polarized abelian varieties (e.g., jacobians) of arbitrary dimension. This particularly applies to the supersingular isogeny walks on which SIDH and Charles, Goren and Lauter's hash function are based. In fact, in a follow-up paper~\cite[\S6.2]{CGLgenusg} the latter authors already hint at the possibility of a higher-dimensional analogue of their hash function.  Last year Takashima~\cite[\S4.2]{takashima} made the concrete proposal of using jacobians of supersingular genus-$2$ curves and their $15$ outgoing $(2,2)$-isogenies, which can be evaluated efficiently through Richelot's formulas. By disallowing backtracking
 he uses this to process one base-$14$ digit for each isogeny evaluation.
 Moreover he provides specific starting curves, such as $y^2 = x^5 + 1$ over $\FF_p$ with $p \equiv 4 \bmod 5$, which allow for all computations to be done over $\FF_{p^2}$, as was shown by himself and Yoshida about a decade ago~\cite{richelotchaining}. Unfortunately Takashima's hash function is not collision-resistant due to the inherent presence of small cycles in the resulting isogeny graph, as was pointed out very recently by Flynn and Ti~\cite{flynnti}, who then proceeded with studying a genus-$2$ variant of SIDH. 
 
The contributions of this paper are as follows. First, in Section~\ref{sec:ss_vs_ss} we argue that the full supersingular isogeny graph is the wrong arena for higher-dimensional analogues of Charles, Goren and Lauter's hash function, and promote the use of superspecial subgraphs. In doing so we give a natural explanation for why Takashima and Yoshida's starting curve  indeed allows for all subsequent arithmetic to be carried out in $\FF_{p^2}$. Second, some first properties of the $(2,2)$-isogeny graph of superspecial principally polarized abelian surfaces are gathered and proved in Section~\ref{sec:22graph} and Appendix~\ref{sec:grobnerproof}. Third and foremost, 
we repair Takashima's hash function by showing that an extremely simple restriction (which still allows us to process one base-$8$ digit, i.e., $3$ bits per isogeny)
both prevents the Flynn--Ti attack and simplifies the reasoning on security; we also show that with high probability, the starting curve $y^2 = (x^2-1)(x^2-2x)(x-1/2)$ over $\FF_p$ with $p \equiv 5 \bmod 6$ naturally avoids running into products of elliptic curves, which as we will see are technical nuisances. The details can be found in Section~\ref{sec:extensions} and Section~\ref{sec:hash}. In Sections~\ref{section:implementation_and_timings} and~\ref{section:remark_improvement} we report on an implementation in Magma and compare its performance with the elliptic curve case of Charles, Goren and Lauter.

\subsection*{Why generalizing?} Besides scientific curiosity, we see a number of motivations for investigating higher-dimensional isogeny-based cryptography:
\begin{enumerate}
    \item There seem to exist some beneficial trade-offs between the larger computational cost of each isogeny evaluation and features such as larger graph sizes, higher numbers of outgoing isogenies, or arithmetic in smaller finite fields. As an illustration of this, we note that in Charles, Goren and Lauter's hash function one needs to compute one square root for each digested bit, while our proposal uses $3$ square roots per $3$ bits, which seems like no improvement at all, except that our square roots are to be extracted in finite fields of about one third of the bit size and can be handled in parallel. See Section~\ref{section:remark_improvement} for some further comments on this.
 \item The fact that higher-dimensional abelian varieties have torsion subgroups of larger rank may allow for a symmetric set-up of SIDH in which Alice and Bob sample their secrets from the same space (but this is not touched upon in the current paper).
\end{enumerate}



\section{
    Supersingular versus superspecial
}
\label{sec:ss_vs_ss}

One apparent point of concern is that in the case of elliptic curves over a finite field of characteristic $p$, supersingularity has many equivalent characterizations whose natural generalizations to higher dimension become distinct notions. 
For instance, one such characterization reads that the trace $t$ of Frobenius satisfies $t \equiv 0 \bmod p$,
which naturally generalizes to the notion of \emph{superspeciality}.\footnote{In the case of jacobians, superspeciality amounts to the Hasse-Witt matrix $M \in \FF_p^{g \times g}$ being zero, where we note that $M \equiv t \bmod p$ when $g = 1$. For arbitrary abelian varieties $A$ being superspecial means that Frobenius acts as the zero map on $H^1(A, \mathcal{O}_A)$.} An alternative characterization states that the Newton polygon is a straight line segment with slope $1/2$; this property makes sense in arbitrary dimension where it is still called \emph{supersingularity}, but in dimension $g \geq 2$ this is a weaker condition than superspeciality. A third characterization is that there exists no non-trivial $p$-torsion. This also makes sense in arbitrary dimension but in dimension $g \geq 3$ it weakens the notion of supersingularity. A curve is called superspecial or supersingular if its accompanying jacobian is superspecial or supersingular respectively.

We refer to Brock's thesis~\cite{brockthesis} and the references therein for some general facts on supersingularity and superspeciality. Most notably, it can be shown that an abelian variety is supersingular if and only if it is isogenous to a product of supersingular elliptic curves, while it is superspecial if and only if it is isomorphic to such a product; moreover in dimension $g \geq 2$ all such products are isomorphic to each other, see e.g.~\cite[Thm.\,2.1A]{brockthesis} or~\cite[p.\,13]{lioort}. Here, isogenous and isomorphic should be understood in the context of abstract abelian varieties, regardless of the principal polarization with which they may come equipped: statements like `all superspecial curves of genus $g \geq 2$ have isomorphic jacobians' are of course not true in general (see Theorem~\ref{size} below for a precise count in case $g=2$).
We will abbreviate principal polarization to p.p.\ from now on and will also assume that a product of elliptic curves always comes with the product polarization, unless stated otherwise.

We believe that the full graph of supersingular p.p.\ abelian varieties is the wrong context in which to study Charles--Goren--Lauter
hash functions in dimension $g \geq 2$. Instead we argue for use of the superspecial
subgraph. Indeed, the family of supersingular p.p.\ abelian varieties over $\FFbar_p$ is
infinite-dimensional, whereas the superspecial subfamily is
$0$-dimensional. The latter implies that there is only a finite number of them and, furthermore, they all admit a model over $\FF_{p^2}$ whose Frobenius endomorphism has characteristic polynomial $\chi(t) = (t \pm p)^{2g}$, in particular it acts 
 as multiplication by $\pm p$; see \cite{ibukiyama}. Assuming that $p$ is odd, this implies that all $2$-torsion is $\FF_{p^2}$-rational, hence so are all $(2,2, \ldots, 2)$-isogenies and their codomains. By~\cite[Lem.\,2.2A]{brockthesis} these are again superspecial p.p.\ abelian varieties whose Frobenius has the same characteristic polynomial, so the argument repeats and we conclude that the full superspecial $(2,2, \ldots, 2)$-isogeny graph is defined over $\FF_{p^2}$. This explains the aforementioned observation by Takashima and Yoshida, whose starting curves are indeed superspecial. See~\cite{oort}, where several more examples of superspecial genus-$2$ curves over $\FF_p$ can be found. These include 
$y^2 = x^5-x$ which is superspecial if and only if $p\equiv 5$ or $7\bmod 8$, and $y^2 = (x^2-1)(x^2-2x)(x-1/2)$ which is superspecial if and only if $p \equiv 5 \bmod 6$. In characteristics $2$ and $3$ superspecial genus-$2$ curves do not exist. In general it seems unknown how to write down the equation of a random superspecial genus-$2$ curve.
 
 Note that superspecial p.p.\ abelian varieties were also considered in Charles, Goren and Lauter's follow-up paper~\cite{CGLgenusg}, albeit in a more theoretical context and using different edge and vertex sets for the associated graphs.

\section{
    Further preliminaries
}

\subsection{Hyperelliptic curves of genus-2}

Let $K$ be a field of characteristic $p>5$. A \textit{(hyperelliptic) curve of genus-2 over $K$} is an algebraic curve defined by an equation of the form $y^2 = f(x)$, where $f(x)\in K[x]$ is a squarefree polynomial of degree 5 or 6. Up to $\overline K$-isormorphism, any genus-2 curve has a representation with a monic polynomial of degree 6 and we will mostly work with these representations since it eases up the notation quite a bit. All formulas provided still work with a degree 5 polynomial if one sees the missing linear factor as `$0\cdot x+1$'. A genus-2 curve is determined (up to $\overline K$-isomorphism) by its Cardona--Quer invariants. 
The specific formulas for these invariants are discussed in \cite{g2invars}, but for our purposes it suffices to know that they consist of an ordered triple $(j_1,j_2,j_3)\in K^3$.

\subsection{Richelot isogenies}

A \textit{Richelot isogeny} is a $(2,2)$-isogeny between jacobians of
genus-2 curves, i.e. the kernel of the isogeny is a group isomorphic to
$\ZZ/2\ZZ\oplus\ZZ/2\ZZ$ that is maximal isotropic with regards to the
\(2\)-Weil pairing. The $2$-torsion of the jacobian of the genus-2 curve $C:
y^2 = f(x) = \prod_{i=1}^6(x-\alpha_i)$ is $\{0\} \cup \{
    [(\alpha_i,0)-(\alpha_j,0)] : i < j \}$, where the square brackets denote linear equivalence classes of divisors. A subgroup of the
2-torsion being maximal isotropic with regards to the \(2\)-Weil pairing in this context simply means that the group contains exactly 3 non-trivial elements such that all $\alpha_i$, $1\leq i\leq 6$, occur exactly once in all the representations combined. Hence the Richelot isogenies can be represented by sets of quadratic factors of $f(x)$ that are pairwise coprime. More precisely, if we define
\[
\begin{cases}
G_1 = g_{1,3}x^2 + g_{1,2}x + g_{1,1} = (x-\alpha_1)(x-\alpha_2),\\
G_2 = g_{2,3}x^2 + g_{2,2}x + g_{2,1} = (x-\alpha_3)(x-\alpha_4),\\
G_3 = g_{3,3}x^2 + g_{3,2}x + g_{3,1} = (x-\alpha_5)(x-\alpha_6),
\end{cases}
\]
then the $(2,2)$-isogeny with kernel $\{0, [(\alpha_1,0)-(\alpha_2,0)]$, $[(\alpha_3,0)-(\alpha_4,0)], [(\alpha_5,0)-(\alpha_6,0)]\}$ can be identified by the \textit{quadratic splitting} $\{G_1, G_2, G_3\}$.
While the above equalities force our quadratic factors to be monic, i.e.\ $g_{i,3}=1$ for all $i$, we incorporate the leading coefficients in our discussion for the sake of generality (e.g., to cope with the degree $5$ case where one of the $g_{i,3}$'s becomes zero).
In any case quadratic splittings are only identified up to permutation
and constant multiples of the three quadratics.  

There are 15 possible ways of organizing the roots $\alpha_i$ into
distinct quadratic splittings. It is possible that the resulting
quadratics
are only defined over an extension of the field over which our curve $C$
is defined, in which case both the corresponding $(2,2)$-isogeny and its codomain also
might be defined over this field extension.  Nevertheless,
if the splitting is fixed by Frobenius \emph{as a set}, then the isogeny
and codomain are defined over the ground field.
As mentioned in Section~\ref{sec:ss_vs_ss}, in the case of superspecial
p.p.\ abelian surfaces, all domains, kernels, $(2,2)$-isogenies and
associated codomains are defined over $\FF_{p^2}$ up to isomorphism.

\begin{proposition}
    \label{prop:Richelot}
    Let \(C: y^2 = G_1(x)\cdot G_2(x)\cdot G_3(x)\)
    be a genus-2 curve,
    with \(\{G_1, G_2, G_3\}\) the quadratic splitting
    associated with a maximal 2-Weil-isotropic subgroup
    \(S \subset J_C[2]\),
    and let \(\phi: J_C \to A \cong J_C/S\)
    be the quotient \((2,2)\)-isogeny.
    Following the notation above, let
    \[
        \delta 
        := 
        \det
        \begin{pmatrix}
            g_{1,3} & g_{1,2} & g_{1,1}\\
            g_{2,3} & g_{2,2} & g_{2,1}\\
            g_{3,3} & g_{3,2} & g_{3,1}
        \end{pmatrix}
        \,.
    \]
    \begin{enumerate}
        \item
            \underline{If \(\delta \not= 0\)}, then
            \(A\) is isomorphic to the jacobian of the genus-2 curve
            \[
                C' : y^2 = \delta^{-1} H_1(x)\cdot H_2(x)\cdot H_3(x)
            \]
            where
            \begin{align*}
                \hspace{1.2cm}
                H_1 & := G'_2G_3-G_2G'_3
                \,,
                &
                H_2 & := G'_3G_1-G_3G'_1
                \,,
                &
                H_3 & := G'_1G_2-G_1G'_2
                \,,
            \end{align*}
            where \(G_i'\) is the derivative of \(G_i\) with respect to \(x\).
            Moreover,
            \(\{H_1,H_2,H_3\}\)
            is a quadratic splitting corresponding 
            to the dual isogeny \(\hat\phi: J_{C'} \to J_C\).
        \item
            \underline{If \(\delta = 0\)}, then
            \(A\) is isomorphic to 
            a product of elliptic curves $E_1\times E_2$. 
            The vanishing of the determinant \(\delta\)
            implies that 
            there exist \(s_1\) and \(s_2\) in \(\FF_{p^2}\)
            such that 
            \[
                G_i = a_{i,1}(x-s_1)^2+a_{i,2}(x-s_2)^2
            \]
            for some \(a_{i,1}\) and \(a_{i,2}\) in \(\FF_{p^2}\)
            for \(i = 1, 2, 3\).
            The elliptic curves forming the product isomorphic to \(A\)
            can be defined by the equations
            \begin{align*}
                E_1 : y^2 & = \prod_{i=1}^3 (a_{i,1}x+a_{i,2})
                \,,
                &
                E_2 : y^2 & = \prod_{i=1}^3 (a_{i,1}+a_{i,2}x)
                \,,
            \end{align*}
            and the isogeny \(\phi\) is induced by \(\phi_1\times\phi_2\),
            where \(\phi_1: C \to E_1\)
            is \((x,y) \mapsto ((x - s_1)^2/(x - s_2)^2,y/(x - s_2)^3)\)
            and \(\phi_1: C \to E_2\)
            is \((x,y) \mapsto ((x - s_2)^2/(x - s_1)^2,y/(x - s_1)^3)\).
    \end{enumerate}
\end{proposition}
For a proof of this proposition and a more in-depth discussion about Richelot isogenies, see
\cite[Chapter~8]{smiththesis}.

\subsection{(2,2)-isogenies from products of elliptic curves}

Consider the p.p.\ abelian surface $E_1 \times E_2$ given by the equations
\begin{align*}
                E_1 : y^2 & = \prod_{i=1}^3 (x-\alpha_i)
                \,,
                &
                E_2 : y^2 & = \prod_{i=1}^3 (x-\beta_i)
                \,.
\end{align*}
Just as in the case of jacobians of genus-2 curves, there are 15 outgoing $(2,2)$-isogenies with domain $E_1\times E_2$. Of these, 9 correspond to an isogeny that is the product of $2$-isogenies on the respective elliptic curves, such that the image of this isogeny is again simply a product of elliptic curves. The other 6 determine an isogeny where the kernel is given by
\[
\kappa = \{(\mathcal O_{E_1}, \mathcal O_{E_2}), (P_1, Q_{\sigma(1)}), (P_2, Q_{\sigma(2)}), (P_3, Q_{\sigma(3)}) \},
\]
with $\mathcal O_{E_1}$ and $\mathcal O_{E_2}$ the neutral element of $E_1$, respectively $E_2$, $\sigma$ a permutation of $\{1,2,3\}$, and $P_i = (\alpha_i,0)$, $Q_i = (\beta_i, 0)$. 
As long as $\kappa$ is not the restriction of the graph of an isomorphism $E_1\rightarrow E_2$, the image of the isogeny determined by $\kappa$ is the jacobian of a genus-2 curve which can be constructed as follows. Define $\Delta_\alpha$ and $\Delta_\beta$ as the discriminants of the monic cubic polynomials $\prod_{i=1}^3 (x-\alpha_i)$ and $\prod_{i=1}^3 (x-\beta_i)$ respectively, and
\begin{align*}
& a_1  = (\alpha_3-\alpha_2)^2/(\beta_3-\beta_2) + (\alpha_2-\alpha_1)^2/(\beta_2-\beta_1) + (\alpha_1-\alpha_3)^2/(\beta_1-\beta_3), \\
& b_1  = (\beta_3-\beta_2)^2/(\alpha_3-\alpha_2) + (\beta_2-\beta_1)^2/(\alpha_2-\alpha_1) + (\beta_1-\beta_3)^2/(\alpha_1-\alpha_3), \\
& a_2  = \alpha_1(\beta_3-\beta_2) + \alpha_2(\beta_1-\beta_3) + \alpha_3(\beta_2-\beta_1), \\
& b_2  = \beta_1(\alpha_3-\alpha_2) + \beta_2(\alpha_1-\alpha_3) + \beta_3(\alpha_2-\alpha_1).
\end{align*}
It can be proved that $\Delta_\alpha, \Delta_\beta, a_1, b_1, a_2, b_2$ are all nonzero, such that $A = \Delta_\beta a_1/a_2$ and $B = \Delta_\alpha b_1/b_2$ are well defined and nonzero as well. With these notations in mind, the image of the $(2,2)$-isogeny with kernel $\kappa$ is the jacobian of the genus-2 curve given by the equation
\begin{multline*}\label{elltojac}
    y^2 = -\left(A(\alpha_2-\alpha_1)(\alpha_1-\alpha_3)x^2 + B(\beta_2-\beta_1)(\beta_1-\beta_3)\right)\\
    \cdot \left(A(\alpha_3-\alpha_2)(\alpha_2-\alpha_1)x^2 + B(\beta_3-\beta_2)(\beta_2-\beta_1)\right)\\
    \cdot \left(A(\alpha_1-\alpha_3)(\alpha_3-\alpha_2)x^2 + B(\beta_1-\beta_3)(\beta_3-\beta_2)\right).
\end{multline*}
The three factors on the right hand side constitute a quadratic splitting for the dual isogeny back to $E_1 \times E_2$; note in particular that these factors are multiples of each other so that the corresponding value of $\delta$ is indeed $0$.

If $E_1\cong E_2$ we will have strictly fewer than six $(2,2)$-isogenies
from $E_1\times E_2$ to the jacobian of a genus-2 curve. The exact
number in this case is given by the formula $6-\#\text{Aut}(E_1)/2$. If the $j$-invariant of $E_1$ is $0$ or $1728$ then this expression is 3 respectively 4 (under the assumption that $p>3$). In all other cases this expression is $5$ since the only automorphisms are $\pm1$.

The final case to consider is when we want to construct an isogeny with domain an abelian surface of the form $E_1\times E_2$, with $E_1\cong E_2$, and of which the kernel $\kappa$ is the restriction of the graph of an isomorphism $\alpha: E_1\rightarrow E_2$. The codomain is then the same as the domain and the $(2,2)$-isogeny is given by
\begin{eqnarray*}
\phi : & E_1\times E_2 & \rightarrow E_1\times E_2\\
 & (P,Q) & \mapsto (P+\hat\alpha(Q), -Q+\alpha(P)),
\end{eqnarray*}
which is clearly self-dual.

For a proof of these previous statements and a more in-depth discussion, see \cite{degenformulas}, \cite{kani} and \cite{bruin}.

\section{
    The superspecial \texorpdfstring{\((2,2)\)}{(2,2)}-isogeny graph
}
\label{sec:22graph}

For each prime \(p\),
we define a directed multigraph $\thegraph_p$ as
follows.\footnote{%
    Every \((2,2)\)-isogeny $\phi: A_1\to A_2$ has a unique dual 
    \((2,2)\)-isogeny $\hat\phi: A_2\to A_1$, so one might think that
    we could easily treat \(\thegraph_p\) as an undirected graph.
    Unfortunately, this may fail if $A_1$ has automorphisms different from $\pm 1$. Indeed, in that case it is possible that two non-isomorphic \((2,2)\)-isogenies 
    $\phi: A_1\to A_2$ and $\psi: A_1\to A_2$
    are obtained from each other by pre-composition with such an automorphism, 
    so that their duals are obtained from one another by \emph{post}-composition with this automorphism (more precisely if
    $\phi = \psi \circ \alpha$ then $\hat{\phi} = \alpha^{-1} \circ \hat{\psi}$). So these duals have the same kernel, hence they are isomorphic. 
    In the elliptic curve case, this technicality can be combated by 
    choosing $p\equiv 1\bmod 12$, since then the automorphisms of all
    curves are always $\pm 1$. In the case of superspecial genus-2 curves, however,
    no such convenient restriction exists:
    there are jacobians with a different number of automorphisms 
    for any prime $p$ \cite{oort}.
}
The vertices of $\thegraph_p$ represent the isomorphism classes of 
superspecial p.p.\ abelian surfaces defined over $\FFbar_p$.
The graph $\thegraph_p$ has an edge from vertex $A_1$ to vertex $A_2$
for every $(2,2)$-isogeny from the superspecial p.p.\ abelian surface
corresponding to $A_1$ to the one corresponding to $A_2$, again up to isomorphism. 
Here, isomorphisms of outgoing \((2,2)\)-isogenies are commutative diagrams
\[
    \begin{tikzcd}
        A_1 \arrow{r}{\phi} \arrow{dr}{\phi'} & A_2 \arrow{d}{\iota} \\
        & A_2' 
    \end{tikzcd}
\]
where \(\phi\) and \(\phi'\) are \((2,2)\)-isogenies 
and \(\iota\) is an isomorphism 
of superspecial p.p.\ abelian surfaces.
Since the isomorphism class of an outgoing isogeny is uniquely determined by its kernel, this simply means that we have an outgoing edge for each $(2,2)$-subgroup of $A_1$, i.e., each subgroup that is isomorphic to
$\ZZ/2\ZZ\oplus\ZZ/2\ZZ$ and maximal isotropic with regards to the
\(2\)-Weil pairing.

By construction, $\thegraph_p$ is a 15-regular (multi)graph,
since both types of superspecial p.p.\ abelian surfaces have 
15 different $(2,2)$-isogenies. 
One might simplify the situation by combining parallel edges
to turn $\thegraph_p$ into a simple directed graph, but 
for our application we will need to distinguish between all $15$ outgoing edges.
In any case, for large $p$ the number of parallel edges is expected to
be negligible relative to the size of the graph
(for very small \(p\), where there are few superspecial p.p.\ abelian surfaces,
the opposite holds---as we will see in \S\ref{sec:G13}).

The vertices of \(\thegraph_p\) fall into two classes:
\[
    V(\thegraph_p) = \thegraphell_p + \thegraphjac_p
    \,,
\]
where \(\thegraphell_p\) is the set of isomorphism classes corresponding to
products of supersingular elliptic curves, and \(\thegraphjac_p\) is the set of
isomorphism classes of superspecial genus-2 jacobians.
Theorem~\ref{size} gives us the cardinalities of these subsets.

\begin{theorem}\label{size}
    Let \(\thegraph_p\), \(\thegraphell_p\), and \(\thegraphjac_p\)
    be defined as above.
    \begin{itemize}
        \item
            If \(p = 2\) or \(3\),
            then \(\#\thegraphjac_p = 0\) and \(\#\thegraphell_p = 1\).
        \item
            If \(p = 5\),
            then \(\#\thegraphjac_p = 1\) and \(\#\thegraphell_p = 1\).
        \item
            If \(p > 5\), 
            then
            \begin{align*}
                \#\thegraphjac_p & = \frac{p^3+24p^2+141p-346}{2880} + \delta_p
                \intertext{and}
                \#\thegraphell_p & = 
                \frac{1}{2}\left( \frac{p-1}{12} + \epsilon_p  \right) \left( \frac{p-1}{12} + \epsilon_p +1 \right)
                \,,
            \end{align*}
            where $\delta_p\in [0,\frac{881}{720}]$ 
            depends only on $p\bmod{120}$
            and $\epsilon_p \in \left[0,\frac{7}{6} \right]$ 
            depends only on $p\bmod{12}$.
    \end{itemize}
\end{theorem}
\begin{proof}
    The values for \(\#\thegraphjac_p\) appear in 
    \cite[Theorem 3.10(b)]{brockthesis} or \cite[Theorem 3.3]{oort}.
    The formulas for \(\#\thegraphell_p\)
    follow from the fact that up to \(\FFbar_p\)-isomorphism,
    the number of supersingular elliptic curves over $\FFbar_p$
    is $(p-1)/12 + \epsilon_p$,
    where $\epsilon_p \in \left[0,\frac{7}{6} \right]$  depends only on $p\bmod{12}$
    (see for example \cite[Section V, Theorem 4.1(c)]{silverman2}).
\end{proof}

%
%

Theorem~\ref{size} implies that $\thegraph_p$ is a finite graph,
although this could already be derived from the fact that every
isomorphism class of superspecial p.p.\ abelian surfaces has a
representative defined over $\FF_{p^2}$. 
Asymptotically, we have
\begin{align*}
    \#\thegraph_p & = O(p^3) \,,
    &
    \#\thegraphell_p & = O(p^2) \,,
    &
    \#\thegraphjac_p & = O(p^3) \,.
\end{align*}
In particular, the
proportion of superspecial p.p.\ abelian surfaces that are the product
of 2 supersingular elliptic curves is $O(1/p)$ relative to the total
size of the graph: for $p$ large, the number of vertices
in \(\thegraph_p\) that are \textit{not} in \(\thegraphjac_p\) is negligible.

Informally, when $p$ is large, one could see \(\thegraphell_p\) as the `boundary'
of the graph $\thegraph_p$, and \(\thegraphjac_p\) as the `interior'. A first
reason is the size argument we just made. A second reason is the
connectivity of the 2 types of superspecial p.p.\ abelian surfaces that
we briefly touched on in the preliminaries. Indeed, every product of
elliptic curves has at least 9 out of 15 $(2,2)$-isogenies that have a
codomain that is a product of elliptic curves as well, hence this part
of our graph is very well connected while only making up a fraction of
our graph. Vice versa there is also no jacobian of a genus-2 curve that
could be `hiding' in between the products of elliptic curves, which we
can make precise with the following theorem.

\begin{theorem}
    \label{th:graph}
    \label{grobnertheorem}
    With the notation above:
    \begin{enumerate}
        \item
            Suppose \(p \not= 5\). If \(J\) is a vertex in \(\thegraphjac_p\subset\thegraph_p\),
            then (counting multiplicity) at most 6 of the 15 edges out of \(J\)
            are to vertices in \(\thegraphell_p\).
        \item
            If \(E\) is a vertex in \(\thegraphell_p\subset\thegraph_p\),
            then (counting multiplicity) at most 6 of the 15 edges out of \(E\)
            are to vertices in \(\thegraphjac_p\).
    \end{enumerate}
\end{theorem}
\begin{proof} 
The second part of this theorem was mentioned in the preliminaries and it follows from the fact that 9 out of 15 $(2,2)$-isogenies are simply a product of $2$-isogenies on the separate elliptic curve factors. A proof of a more general formula can be found in \cite{kani}. For a proof of the first part, see Appendix~\ref{sec:grobnerproof}.
\end{proof}

A simple counting argument then tells us that for sufficiently large $p$,
the chance of a vertex in \(\thegraphjac_p\) having a neighbour in \(\thegraphell_p\) in
our graph $\thegraph_p$ is negligible. Intuitively this makes sense,
since the $\delta$ in Proposition \ref{prop:Richelot} is the
determinant of a seemingly random $3\times3$ matrix for large $p$, and
will therefore almost surely be nonzero.

We now state a pair of conjectures
inspired by analogous theorems for the elliptic supersingular
\(2\)-isogeny graph.

\begin{conjecture}
    \label{conj:connected}
    The graph $\thegraph_p$ is connected.
\end{conjecture}

Conjecture~\ref{conj:connected} is the most natural,
but will not necessarily be the one that we will base choices on.
We mainly state it due to the analogy with the elliptic curve case.

\begin{conjecture}
    \label{conj:jacobian-connected}
    The subgraph of \(\thegraph_p\) supported on \(\thegraphjac_p\)
    is connected.
\end{conjecture}

Conjecture~\ref{conj:jacobian-connected} (which is identical
to Conjecture~\ref{conj:connected} in the elliptic case)
is more relevant to our discussion. 
It implies that \(\thegraphell_p\) not only makes no significant
contribution to the size of \(\thegraph_p\) as \(p\to\infty\),
but it is also not essential for connectivity.
(Thus, we consider \(\thegraphell_p\) to be the ``boundary'' of $\thegraph_p$.)
Conjecture~\ref{conj:jacobian-connected} implies
Conjecture~\ref{conj:connected},
since every vertex in \(\thegraphell_p\)
has at least 4 outgoing edges into \(\thegraphjac_p\) for $p>3$ 
(as mentioned in the preliminaries).
As a final note, one may wonder if all non-superspecial supersingular p.p.\ abelian surfaces also form a similar connected component (which is necessarily infinite). Given that we will not make use of these abelian surfaces, we will not explore that thought any further.

\section{
    The graph \texorpdfstring{$\thegraph_{13}$}{G13}
}
\label{sec:G13}

We now give a small example to show the possible case distinctions that
can occur in the graphs $\thegraph_p$. We take \(p = 13\), since this
yields a small graph that still exhibits most of the subtleties and
pathologies that one might encounter in larger graphs.

Figure \ref{fig:G13} shows~\(\thegraph_{13}\).
There are 3 superspecial genus-2 curves defined over $\FFbar_{13}$
up to isomorphism, say \(C_i\) for \(i\) in \(\{1,2,3\}\); we denote
their jacobians by $J_{C_i}$. There is only 1 supersingular elliptic
curve defined over $\FFbar_{13}$ up to isomorphism, say $E$, so there
is only one vertex in \(\thegraph_{13}\) that corresponds to 
a product of elliptic curves.

First of all it is easily verifiable that there are at most 6 outgoing
edges from any $J_{C_i}$ to $E\times E$, see Appendix~\ref{sec:grobnerproof}. Furthermore, since clearly
$E\cong E$, there are strictly fewer than 6 outgoing edges from $E\times E$ to jacobians of genus-2 curves. Since the $j$-invariant of $E$ is not in $\{0,1728\}$, we know there are exactly 5 like that, so the remaining 10 edges must go to products of elliptic curves as well, which here (by lack of other options) means a loop with multiplicity 10.

This example also shows clearly why direction is important in the graph.
There are 4 edges from $J_{C_1}$ to $J_{C_2}$, but only 1 edge back.
In other words $C_1 : y^2 = x^5-x$ has 4 quadratic
splittings whose associated Richelot isogenies 
have $J_{C_2}$ as codomain,\footnote{Up to isomorphism, that is: the resulting equations for the curve $C_2$ are in fact different, but the Cardona--Quer invariants are the same.} while starting from any Weierstra\ss\ equation for $C_2$, only one quadratic splitting gives rise to a Richelot isogeny with $J_{C_1}$ as codomain. 
This stems from the fact that the 4 corresponding $(2,2)$-subgroups of $J_{C_1}$ are mapped to each other by an automorphism of $J_{C_1}$. In other words the 4 resulting isogenies 
\[ \phi_1, \ldots, \phi_4 : J_{C_1} \rightarrow J_{C_2} \]
are obtained from one another by pre-composition with such an automorphism. But then their duals
\[ \hat{\phi_1}, \ldots, \hat{\phi_4} : J_{C_2} \rightarrow J_{C_1} \]
are obtained from each other by \emph{post}-composition with an automorphism. In particular they have the same kernel or, equivalently, they correspond to the same quadratic splitting.


The only thing missing from the graph is a vertex corresponding to a
product of non-isomorphic elliptic curves. Such a vertex always has
9 outgoing edges (possibly loops) to other vertices in \(\thegraphell_p\),
and 6 outgoing edges to vertices in \(\thegraphjac_p\).
The smallest example of this phenomenon is in the graph $\thegraph_{17}$, which
already has double the number of vertices of \(\thegraph_{13}\).

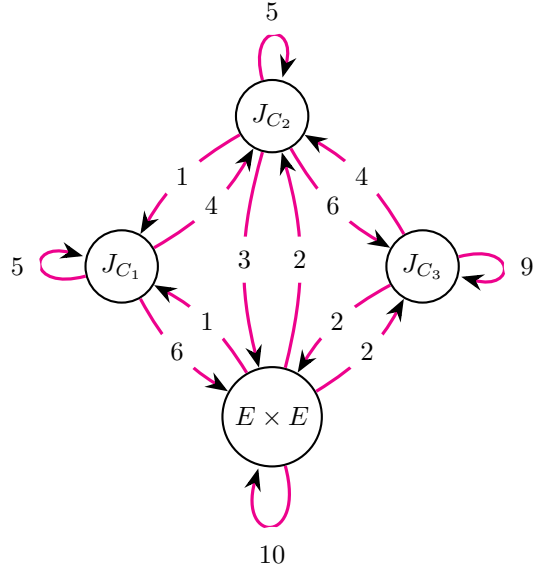
\begin{figure}
\centering
\begin{tikzpicture}
\begin{scope}[every node/.style={circle,thick,draw}]
    \node (A) at (0,0) {$J_{C_1}$};
    \node (B) at (2,2) {$J_{C_2}$};
    \node (C) at (4,0) {$J_{C_3}$};
    \node (D) at (2,-2) {$E\times E$};
\end{scope}

\begin{scope}[>={Stealth[black]},
			  every node/.style={fill=white,circle},
              every edge/.style={draw=magenta,very thick}]
\path [->] (A) edge [loop left] node {$5$} (A);
\path [->] (A) edge [bend right=15] node {$4$} (B);
\path [->] (A) edge [bend right=15] node {$6$} (D);

\path [->] (B) edge [bend right=15] node {$1$} (A);
\path [->] (B) edge [bend right=15] node {$6$} (C);
\path [->] (B) edge [loop above] node {$5$} (B);
\path [->] (B) edge [bend right=15] node {$3$} (D);

\path [->] (C) edge [bend right=15] node {$4$} (B);
\path [->] (C) edge [loop right] node {$9$} (C);
\path [->] (C) edge [bend right=15] node {$2$} (D);

\path [->] (D) edge [loop below] node {$10$} (D);
\path [->] (D) edge [bend right=15] node {$1$} (A);
\path [->] (D) edge [bend right=15] node {$2$} (B);
\path [->] (D) edge [bend right=15] node {$2$} (C);

\end{scope}
\end{tikzpicture}
\caption{The graph $\thegraph_{13}$. The vertices $J_{C_i}$, $i\in\{1,2,3\}$, correspond to jacobians of genus-2 curves, whereas the vertex $E\times E$ corresponds to a product of elliptic curves. The numbers indicate the multiplicities of the edges.}
\label{fig:G13}
\end{figure}

\section{
    A special class of paths in \texorpdfstring{\(\thegraph_p\)}{Gp}
}
\label{sec:extensions}

We are interested in the kinds of isogenies that are represented by
paths in \(\thegraph_p\):
that is, the compositions of isogenies corresponding to adjacent edges.

First, fix a single edge \(\phi_0: A_0 \to A_1\) in \(\thegraph_p\).
By definition, 
\(\phi_0\) represents (up to isomorphism) a \((2,2)\)-isogeny:
that is, an isogeny whose kernel is a maximal \(2\)-Weil isotropic subgroup of
\(A_0[2]\), hence isomorphic to \((\ZZ/2\ZZ)^2\).

Now, consider the set of edges leaving \(A_1\):
these correspond to \((2,2)\)-isogenies that may be composed with \(\phi_0\).
We know that (counting multiplicity) there are fifteen such edges.
These edges fall naturally into three classes \emph{relative to} \(\phi_0\),
according to the structure of the kernel of the composed isogeny
(which, in each case, is a maximal \(4\)-Weil isotropic subgroup of \(A_0[4]\)).

\begin{definition}
    Let \(\phi_0: A_0\to A_1\) and \(\phi_1: A_1\to A_2\)
    be edges in \(\thegraph_p\).
    \begin{itemize}
        \item
            We say that \(\phi_1\) is the (necessarily unique)
            \textbf{dual extension} of \(\phi_0\)
            if \(\ker(\phi_1\circ\phi_0) \cong (\ZZ/2\ZZ)^4\),
            so
            \(\phi_1\circ\phi_0\) is a \((2,2,2,2)\)-isogeny
            (hence isomorphic to \( [2]_{A_0}\)).
            In this case, \(\ker\phi_1 = \phi_0(A_0[2])\).
        \item
            We say that \(\phi_1\) is a \textbf{bad extension} of \(\phi_0\)
            if 
            \(\ker(\phi_1\circ\phi_0) \cong (\ZZ/4\ZZ)\times(\ZZ/2\ZZ)^2\),
            so \(\phi_1\circ\phi_0\) is a \((4,2,2)\)-isogeny.
            In this case
            \((\ker\phi_1)\cap \phi_0(A_0[2]) \cong \ZZ/2\ZZ\),
            so there are precisely \(6\) bad extensions of any given \(\phi_0\).
        \item
            We say that \(\phi_1\) is a \textbf{good extension} of \(\phi_0\)
            if \(\ker(\phi_1\circ\phi_0) \cong (\ZZ/4\ZZ)^2\),
            so \(\phi_1\circ\phi_0\) is a \((4,4)\)-isogeny.
            In this case \((\ker\phi_1)\cap \phi_0(A_0[2]) = 0\),
            so there are precisely \(8\) good extensions of any \(\phi_0\).
    \end{itemize}
\end{definition}

\begin{remark}
    In~\cite[Definition 9.2.1]{smiththesis},
    \emph{good} extensions are called \emph{cyclic} and \emph{bad}
    extensions are called \emph{acyclic}.
    We prefer the good/bad terminology here to avoid confusion with the
    notion of composing isogenies to form eventual cycles in
    \(\thegraph_p\);
    the reason why good is good and bad is bad will become clear
    in Section~\ref{sec:hash}.
\end{remark}

We have seen how the three kinds of extensions
\[
    A_0 \stackrel{\phi_1}{\longrightarrow} A_1 \stackrel{\phi_2}{\longrightarrow} A_2
\]
can be distinguished by how the kernel of \(\phi_2\) intersects
with the image of \(A_0[2]\) under \(\phi_1\).
We can make these criteria more explicit
in terms of the Richelot isogeny formulas.

\subsection{Extensions of isogenies from \texorpdfstring{\(\thegraphjac_p\)}{Jp} to \texorpdfstring{\(\thegraphjac_p\)}{Jp} }

Recall the construction of Richelot isogenies
\(\phi_1: J_{C_0}\to J_{C_1}\)
from Proposition~\ref{prop:Richelot}:
given the curve \(C_0: y^2 = G_1\cdot G_2\cdot G_3\),
we set
\begin{align*}
    H_1 & := G_2'G_3-G_3'G_2
    \,,
    &
    H_2 & := G_3'G_1-G_1'G_3
    \,,
    &
    H_3 & := G_1'G_2-G_2'G_1
    \,.
\end{align*}
The curve \(C_1\) is defined by
\(C_1: y^2 = \delta^{-1}\cdot H_1 \cdot H_2 \cdot H_3\)
where \(\delta := \det(G_1,G_2,G_3)\).
The kernel of \(\phi_1\) corresponds to \(\{G_1,G_2,G_3\}\),
and the subgroup \(\phi_1(J_{C_0}[2]) \subset J_{C_1}[2]\) 
corresponds to \(\{H_1,H_2,H_3\}\).

\begin{proposition}
    \label{prop:good-extensions}
    With the notation above:
    if 
    \begin{align*}
        H_1 & = L_1\cdot L_2\,,
        &
        H_2 & = L_3\cdot L_4\,,
        &
        H_3 & = L_5\cdot L_6\,,
    \end{align*}
    with the \(L_i\) all linear (except possibly for one constant \(L_i\)
    in the case where \(H_1H_2H_3\) is quintic),
    then the good extensions of \(\phi_1\)
    are the Richelot isogenies with kernels corresponding to one of the
    following factorizations of \(H_1H_2H_3\):
    \begin{align*}
        ( L_1L_3, L_2L_5, L_4L_6 ), \quad
        ( L_1L_3, L_2L_6, L_4L_5 ), \\
        ( L_1L_4, L_2L_5, L_3L_6 ), \quad
        ( L_1L_4, L_2L_6, L_3L_5 ), \\
        ( L_1L_5, L_2L_3, L_4L_6 ), \quad
        ( L_1L_5, L_2L_4, L_3L_6 ), \\
        ( L_1L_6, L_2L_3, L_4L_5 ), \quad
        ( L_1L_6, L_2L_4, L_3L_5 ).
    \end{align*}
\end{proposition}
\begin{proof}
    The 
    quadratic splitting
    \(\{H_1,H_2,H_3\}\)
    corresponds to the subgroup of \(J_{C_1}[2]\)
    which is the kernel of the dual \(\hat\phi_1\),
    and also the image \(\phi_1(J_{C_0}[2])\).
    The good extensions of \(\phi_1\)
    are those whose kernel intersects trivially with
    \(\phi_1(J_{C_2}[2])\);
    they therefore correspond to the quadratic splittings
    with no quadratics proportional to any of the \(H_i\).
    The list of 8 splittings above follows from direct calculation.
\end{proof}

We now discuss the good extensions of isogenies involving products of elliptic curves. This is mainly for the sake of completeness, because in our proposed hash function below, these cases will not be implemented.

\subsection{Extensions of isogenies from \texorpdfstring{\(\thegraphjac_p\)}{Jp} to \texorpdfstring{\(\thegraphell_p\)}{Ep} }

Recall from the preliminaries that for a $(2,2)$-isogeny $\phi_1: J_{C_0}\rightarrow E_1\times E_2$, the domain can be written as the jacobian of a curve $C_0: y^2 = G_1G_2G_3$, where
            \[
                G_i = a_{i,1}(x-s_1)^2+a_{i,2}(x-s_2)^2
            \]
for certain $s_1,s_2,a_{i,1},a_{i,2}\in\FF_{p^2}$ for $i=1,2,3$. The elliptic curves determining the codomain can then be defined by the equations
            \begin{align*}
                E_1 : y^2 & = \prod_{i=1}^3 (a_{i,1}x+a_{i,2})
                \,,
                &
                E_2 : y^2 & = \prod_{i=1}^3 (a_{i,1}+a_{i,2}x)
                \,.
            \end{align*}
For $i=1,2,3$ we will write $\{\alpha_i,\alpha_i'\}$ for the roots of $G_i$, $P_i = (-a_{i,2}/a_{i,1}, 0)$ for the Weierstra\ss\ points of $E_1$, $Q_i = (-a_{i,1}/a_{i,2},0)$ for the Weierstra\ss\ points of $E_2$, and $\mathcal O_{E_1}$ and $\mathcal O_{E_2}$ for the neutral element of respectively $E_1$ and $E_2$.
            
\begin{proposition}\label{prop:JtoE}
    With the notation above, the good extensions of $\phi_1$ are the $(2,2)$-isogenies with kernel one of the 6 combinations
\[
        \{ (\mathcal O_{E_1}, \mathcal O_{E_2}), (P_i, \mathcal O_{E_2}), (\mathcal O_{E_1}, Q_j), (P_i, Q_j)  \},
\]
for $i\neq j$ in $\{1,2,3\}$, or one of
\begin{align*}
    \{ (\mathcal O_{E_1}, \mathcal O_{E_2}), (P_1,Q_2), (P_2,Q_3), (P_3,Q_1) \},\\
    \{ (\mathcal O_{E_1}, \mathcal O_{E_2}), (P_1,Q_3), (P_2,Q_1), (P_3,Q_2) \}.
\end{align*}
\end{proposition}

\begin{proof}
The proof of the formulas in \cite[Proposition 8.3.1]{smiththesis} shows that, for $\{i,j,k\}=\{1,2,3\}$, the 2-torsion elements $[(\alpha_i,0)-(\alpha_j,0)]$, $[(\alpha_i,0)-(\alpha_j',0)]$, $[(\alpha_i',0)-(\alpha_j,0)]$, $[(\alpha_i',0)-(\alpha_j',0)]$ get mapped to $(P_k,Q_k)$ in $E_1\times E_2$. So the good extensions of $\phi_1$ are the isogenies whose kernels intersect \[\phi_1(J_{C_0}[2]) = \{ (\mathcal O_{E_1}, \mathcal O_{E_2}), (P_1,Q_1), (P_2,Q_2), (P_3,Q_3) \}\] trivially, which are exactly the ones listed.
\end{proof}
Note that in the previous proposition, the 6 good extensions of the first type always have a product of elliptic curves as codomain. The other 2 will typically be to a jacobian of a genus-2 curve, unless $E_1\cong E_2$ and the given kernel is contained in the graph of an isomorphism $E_1\rightarrow E_2$.

\subsection{Extensions of isogenies from \texorpdfstring{\(\thegraphell_p\)}{Ep} to \texorpdfstring{\(\thegraphjac_p\)}{Jp}}

Recall that every $(2,2)$-isogeny $\phi_1: E_1\times E_2\rightarrow J_{C_1}$, with
\begin{align*}
                E_1 : y^2 & = \prod_{i=1}^3 (x-\alpha_i)
                \,,
                &
                E_2 : y^2 & = \prod_{i=1}^3 (x-\beta_i)
                \,,
\end{align*}
always has codomain the jacobian of a genus-2 curve $C_1$ that can be defined by an equation of the form
\begin{multline}\label{elltojac2}
    y^2 = -\left(A(\alpha_2-\alpha_1)(\alpha_1-\alpha_3)x^2 + B(\beta_2-\beta_1)(\beta_1-\beta_3)\right)\\
    \cdot \left(A(\alpha_3-\alpha_2)(\alpha_2-\alpha_1)x^2 + B(\beta_3-\beta_2)(\beta_2-\beta_1)\right)\\
    \cdot \left(A(\alpha_1-\alpha_3)(\alpha_3-\alpha_2)x^2 + B(\beta_1-\beta_3)(\beta_3-\beta_2)\right),
\end{multline}
up to permutation of the roots $\beta_i$, for well-defined nonzero constants $A$ and $B$ that depend on $\alpha_i$ and $\beta_i$. We will denote the quadratic factors on the right hand side of Equation \ref{elltojac2} on the first, second and third line by $H_1$, $H_2$ and $H_3$ respectively, such that $C_1 : y^2 = - H_1 \cdot H_2 \cdot H_3$.

\begin{proposition}
    \label{prop:good-extensionsEtoJ}
    With the notation above:
    if 
    \begin{align*}
        H_1 & = L_1\cdot L_2\,,
        &
        H_2 & = L_3\cdot L_4\,,
        &
        H_3 & = L_5\cdot L_6\,,
    \end{align*}
    with the \(L_i\) all linear (except possibly for one constant \(L_i\)
    in the case where \(H_1H_2H_3\) is quintic),
    then the good extensions of \(\phi_1\)
    are the Richelot isogenies with kernels corresponding to one of the
    following factorizations of \(H_1H_2H_3\):
    \begin{align*}
        ( L_1L_3, L_2L_5, L_4L_6 ), \quad
        ( L_1L_3, L_2L_6, L_4L_5 ), \\
        ( L_1L_4, L_2L_5, L_3L_6 ), \quad
        ( L_1L_4, L_2L_6, L_3L_5 ), \\
        ( L_1L_5, L_2L_3, L_4L_6 ), \quad
        ( L_1L_5, L_2L_4, L_3L_6 ), \\
        ( L_1L_6, L_2L_3, L_4L_5 ), \quad
        ( L_1L_6, L_2L_4, L_3L_5 ).
    \end{align*}
\end{proposition}
\begin{proof}
The proof of Equation \ref{elltojac2} in \cite{degenformulas} constructs the dual isogeny $\hat\phi_1: J_{C_1}\rightarrow E_1'\times E_2'$, where $E_1'\cong E_1$ and $E_2'\cong E_2$. More specifically, $E_1'$ and $E_2'$ are given by
\begin{multline*}
    E_1': y^2 = -\left(A(\alpha_2-\alpha_1)(\alpha_1-\alpha_3)x + B(\beta_2-\beta_1)(\beta_1-\beta_3)\right)\\
    \cdot \left(A(\alpha_3-\alpha_2)(\alpha_2-\alpha_1)x + B(\beta_3-\beta_2)(\beta_2-\beta_1)\right)\\
    \cdot \left(A(\alpha_1-\alpha_3)(\alpha_3-\alpha_2)x + B(\beta_1-\beta_3)(\beta_3-\beta_2)\right),
\end{multline*}
\begin{multline*}
    E_2': y^2 = -\left(A(\alpha_2-\alpha_1)(\alpha_1-\alpha_3) + B(\beta_2-\beta_1)(\beta_1-\beta_3)x\right)\\
    \cdot \left(A(\alpha_3-\alpha_2)(\alpha_2-\alpha_1) + B(\beta_3-\beta_2)(\beta_2-\beta_1)x\right)\\
    \cdot \left(A(\alpha_1-\alpha_3)(\alpha_3-\alpha_2) + B(\beta_1-\beta_3)(\beta_3-\beta_2)x\right).
\end{multline*}
Hence the quadratic splitting $\{H_1, H_2, H_3\}$ corresponds to the subgroup of $J_{C_1}[2]$ which is the kernel of the dual $\hat\phi_1$ and we can continue the proof just as in the Richelot isogeny case.
\end{proof}

\subsection{Extensions of isogenies from \texorpdfstring{\(\thegraphell_p\)}{Ep} to \texorpdfstring{\(\thegraphell_p\)}{Ep}}

\begin{proposition}\label{prop:EtoE}
Let $\phi_1: E_1\times E_2\rightarrow E_1'\times E_2'$ be a $(2,2)$-isogeny. Denote by $\mathcal O_{E_1}$, $\mathcal O_{E_2}$, $\mathcal O_{E_1'}$, $\mathcal O_{E_2'}$ the identity elements of respectively $E_1$, $E_2$, $E_1'$ and $E_2'$. For $i=1,2,3$ we write $P_i, Q_i, P_i',  Q_i'$ for the Weierstra\ss\ points of respectively $E_1, E_2, E_1', E_2'$. If $$\ker(\phi_1) = \{ (\mathcal O_{E_1}, \mathcal O_{E_2}), (P_1, \mathcal O_{E_2}), (\mathcal O_{E_1}, Q_1), (P_1, Q_1)  \}, $$
and $\phi_1|_{E_1}(P_2) = \phi_1|_{E_1}(P_3) = P_1'$, $\phi_1|_{E_2}(Q_2) = \phi_1|_{E_2}(Q_3) = Q_1'$,  then the good extensions of $\phi_1$ are the isogenies with kernel one of the 4 combinations
\[
        \{ (\mathcal O_{E_1'}, \mathcal O_{E_2'}), (P_i', \mathcal O_{E_2'}), (\mathcal O_{E_1'}, Q_j'), (P_i', Q_j')  \},
\]
where $i\neq 1$ and $j\neq 1$, or one of
\begin{align*}
    \{ (\mathcal O_{E_1'}, \mathcal O_{E_2'}), (P_1',Q_2'), (P_2',Q_3'), (P_3',Q_1') \},\\
    \{ (\mathcal O_{E_1'}, \mathcal O_{E_2'}), (P_1',Q_3'), (P_2',Q_1'), (P_3',Q_2') \},\\
    \{ (\mathcal O_{E_1'}, \mathcal O_{E_2'}), (P_1',Q_2'), (P_2',Q_1'), (P_3',Q_3') \},\\
    \{ (\mathcal O_{E_1'}, \mathcal O_{E_2'}), (P_1',Q_3'), (P_2',Q_2'), (P_3',Q_1') \}.
\end{align*}
\end{proposition}

\begin{proof}
The good extensions are determined by the $(2,2)$-isogenies that intersect
\[
\{ (\mathcal O_{E_1'}, \mathcal O_{E_2'}), (P_1', \mathcal O_{E_2'}), (\mathcal O_{E_1'}, Q_1'), (P_1', Q_1')  \}
\]
trivially, so the proof is immediate.
\end{proof}

\subsection{Connectedness}

\begin{conjecture}
    \label{conj:cyclic-jacobian-connected}
    For every 2 vertices $A$ and $A'$ in \(\thegraphjac_p\subset\thegraph_p\),
    there exists a path
    \[
        A = A_0 
        \xrightarrow{\phi_0} A_1 
        \xrightarrow{\phi_1} \cdots
        \xrightarrow{\phi_{k-1}} A_{k} = A'
    \]
    of \(k\) edges, for some \(k \ge 0\),
    such that all of the \(A_i\) are in \(\thegraphjac_p\)
    and each \(\phi_i, i\neq 0,\) is a good extension of \(\phi_{i-1}\).
    (The composed isogeny is then a \((2^k,2^k)\)-isogeny.)
\end{conjecture}

%

Conjecture~\ref{conj:cyclic-jacobian-connected}
is our strongest conjecture.
It differs from Conjecture~\ref{conj:jacobian-connected} in that at each step
in a path, the number of choices is reduced from all 15 isogenies
to the 8 good isogenies.
Conjecture~\ref{conj:cyclic-jacobian-connected} is easy to verify
for small $p$ using the formulas for
Richelot isogenies and the exact formula from Theorem~\ref{size}. 
We verified this part of the conjecture for $p\leq 1013$ using Magma,
but from then onward the computations become slow since we work with
graphs of several hundred thousands of vertices already.
Nonetheless, this is a first
indication that Conjecture~\ref{conj:cyclic-jacobian-connected} might hold.

\section{
    Hash functions from Richelot isogenies
}
\label{sec:hash}

Turning the graph $\thegraph_p$ into a hash function happens analogously to the elliptic curve case with some small caveats. We will first describe the function, thereby repairing Takashima's proposal from \cite{takashima}, and then work out some small remarks and argue why certain choices were made.

We start by choosing a large prime $p$ (in function of some security
parameter $\lambda$) such that $p\equiv 5\bmod 6$. We start at the
vertex corresponding to the jacobian of the genus-2 curve $C_0$ defined
over $\FF_{p^2}$, given by the equation $y^2 = x(x-1)(x+1)(x-2)(x-1/2)$.
The hash function starts by multiplying the input by $8^{10}$, or equivalently, padding it with 30 zeroes. The hashing will happen 3 bits at a time,
with each three bits determining a choice of one of the eight good
extensions relative to the previous step.
So for our starting vertex we will need to make an initial choice as if we performed a step prior to starting. The quadratic splitting we will choose for $C_0$ is \[ \left\{ x^2-1, x^2-2x, x- \frac{1}{2} \right\}.\]
The 8 quadratic splittings that we will consider are those that have
\textit{no} quadratic factor in common with the one that was obtained from
the previous step. These splittings are then ordered according to some
lexicographical order of the roots. In practice this means we just need to fix a quadratic
equation that determines the field extension $\FF_p \subseteq \FF_{p^2}$. Next we process 3 bits of our input according to the order of the 8 chosen edges. If the chosen edge leads to a vertex corresponding to the product of elliptic curves, the function stops and outputs an error. If the chosen edge leads to a vertex corresponding to a jacobian of a genus-2 curve then we rinse and repeat for the next 3 bits. Once the entire message has been processed we output the Cardona--Quer invariants of the genus-2 curve corresponding to the vertex we ended up in.

\subsection{Avoiding trivial cycles} \label{sec:avoiding}

A hash function should be collision resistant, so we need to at least
avoid trivial cycles in our graph. In the elliptic curve case, this is
simply done by disallowing the edge associated to the dual isogeny from
where we just came. 

Similarly,
we must avoid using dual isogenies
when walking in \(\thegraph_p\),
to avoid extremely easy cycles:
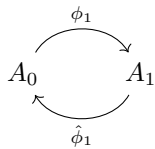
\[
    \begin{tikzcd}
        A_0\arrow[bend left=60]{r}{\phi_1} & A_1 \arrow[bend
        left=60]{l}{\hat\phi_1}
    \end{tikzcd}
\]
But there is an additional subtlety in genus-2,
as noted in~\cite{flynnti}.
If we follow a \((2,2)\)-isogeny \(A_0 \to A_1\)
with a bad extension \(A_1 \to A_2\),
then we end up with a \((4,2,2)\)-isogeny;
but then, for every $(4,2,2)$-isogeny $A_0\rightarrow A_2$ there are 3
distinct ways to split it up into the concatenation of two $(2,2)$-isogenies as in the following diagram.
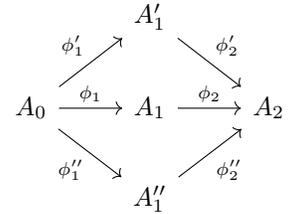
\[
    \begin{tikzcd}
        & A_1' \arrow{dr}{\phi_2'} \\
        A_0 \arrow{ur}{\phi_1'} \arrow{r}{\phi_1} \arrow[swap]{dr}{\phi_1''} 
            & A_1 \arrow{r}{\phi_2} & A_2 \\
        & A_1'' \arrow[swap]{ur}{\phi_2''}
    \end{tikzcd}
\]
Luckily all these cases are easy to distinguish,
as we saw in Section~\ref{sec:extensions}.

The eight $(2,2)$-isogenies corresponding to cyclic extensions do not result in trivial cycles.
In practice this means that, after a choice for our initial $(2,2)$-isogeny corresponding to one of 15 possible edges, we are left with only 8 options at every next step along the way.  
%
%
%
This implies that we should not only keep track of our current vertex by some form of equation, but also by some order of the roots of that equation (or more precisely, by a quadratic splitting). 



This observation means we can hash up to 3 bits at every step in our hash function and that a hash will always correspond to computing a $(2^k,2^k)$-isogeny.

\subsection{Products of elliptic curves}

For our hash function, the vertices corresponding to products of elliptic curves are a nuisance for the following reasons.
\begin{itemize}
    \item There is no clear candidate invariant that is similar to the
        ordered triple in case of the genus-2 Cardona--Quer invariants.
        So ideally, we would prefer not to end the hash function in a 
        vertex like this.
    \item The formulas involving products of elliptic curves are a lot more involved than the Richelot isogenies, and their simplicity was one of the main reasons for the restriction to $(2,2)$-isogenies.
\end{itemize}

In the way we presented our hash function, we simply use Richelot isogenies only and let our hash function break down whenever we pass at a vertex corresponding to a product of elliptic curves. Given that this only occurs with probability $O(1/p)$, this only happens with negligible probability for practical values of $p$.

An alternative way of dealing with this is as follows. Assume we try to process a step in our hash function that corresponds to a $(2,2)$-isogeny between a jacobian of a genus-2 curve and a product of elliptic curves. Then (in the same step) we immediately choose one edge corresponding to a $(2,2)$-isogeny from the product of elliptic curves back to a jacobian of a genus-2 curve. This has to be done in a deterministic way and we should avoid the dual and bad extensions since they would result in small cycles in $\mathcal G_p$. Unfortunately Proposition \ref{prop:JtoE} tells us that we can only find 2 good extensions that possibly have the jacobian of a genus-2 curve as codomain. In the case of $E\times E$, with $E$ having $j$-invariant 0 or 1728, these kernels may both be to a product of elliptic curves again. Solving this issue can be done by either choosing $p\equiv 1\bmod 12$ (such that elliptic curves with $j$-invariant 0 and 1728 never occur), or by (deterministically) using the results from Proposition \ref{prop:EtoE} to add an extra step in this specific case.

A third option is to keep working with all the formulas for products of elliptic curves as well. This means we should find a way to merge the Cardona-Quer invariants and (unordered) pairs of $j$-invariants into one output type, which is only an issue when \textit{ending} in a product of elliptic curves.

\subsection{Initial choices}

As mentioned earlier, there is no known way to generate
the equation of a random superspecial genus-2 curve that is defined over
$\FF_{p^2}$. Some specific examples such as $y^2 = x^5-x$ with $p\equiv 5$ or $7 \bmod 8$ are listed in \cite{oort}. 
Unfortunately, the examples that are easiest to represent all have some $(2,2)$-isogenies with codomain the product of 2 supersingular elliptic curves. This seems to imply that we cannot avoid having to deal with vertices corresponding to products of elliptic curves, instead of ignoring them.

However, another initial choice to make is whether we start by picking one of 15 possible
edges or already restrict ourselves to 8, since this is needed for every
subsequent step anyway. We will take only 8 which
means we need to choose an initial quadratic splitting instead of just
an initial curve.\footnote{Remark that in this case,
Conjecture~\ref{conj:cyclic-jacobian-connected} is no longer strong
enough to prove that we can reach all vertices in $\thegraph_p$, 
since it relies on having all 15 initial $(2,2)$-isogenies present. 
However, there is no clear reason to assume that only allowing 8 out of 15 possible edges for our initial choice all of a sudden would disallow us to reach certain vertices.}
Fortunately this solves our problem of finding an appropriate starting curve in a way.
Consider $C_0$, the genus-2 curve given by $y^2 =
x(x-1)(x+1)(x-2)(x-1/2)$ defined over $\FF_p$ with $p>5$. Then $C_0$ is superspecial if and only if $p\equiv 5\bmod 6$ \cite{oort}. Now the vertex corresponding to the jacobian of $C_0$ has 4 neighbours that are products of supersingular elliptic curves. However, if we take the initial quadratic splitting $\left\{x^2-1, x^2-2x, x- \frac{1}{2} \right\}$,
then the 8 allowed outgoing $(2,2)$-isogenies all have the jacobian of a superspecial genus-2 curve as codomain. The only restriction this puts on our hash function is that we need to work with a prime $p$ such that $p\equiv 5\bmod 6$, but this is easy to enforce.

An issue that arises with this curve $C_0$ however, is that its jacobian has many automorphisms and hence has multiple outgoing isogenies with the same codomain.\footnote{Remark that in the elliptic curve case the same thing happens with for example $y^2 = x^3 + x$ with $p\equiv 3\bmod 4$.} More precisely, starting from the given splitting of $C_0$, the 8 good extensions only have 3 distinct codomains up to isomorphism, one of which even occurs with multiplicity 5, which leads to trivial cycles in our graph. An easy way to fix this is to simply take a (relatively short) deterministic path prior to starting to hash our input, or equivalently, pad the input with some zeroes from the right. For other possible starting curves, this padding can be used to additionally avoid products of elliptic curves. In the end, choosing a starting curve seems to be a choice between either a very compact representation, or not having to do any extra computations that avoid trivial collisions.

\subsection{Security}

The security of our hash function depends on the hardness of finding isogenies between certain p.p.\ abelian surfaces. A lot of the choices discussed in the previous subsections make slight alterations to the underlying mathematical hard problems. We will formulate them in a general form to keep them succinct since we don't think any of the changes would impact the hardness of the problems. In essence they are equivalent to the hard problems from the elliptic curve hash function in \cite{CGL}.

\begin{problem}\label{prob1}
Given two superspecial genus-2 curves $C_1$ and $C_2$ defined over $\FF_{p^2}$, find a $(2^k, 2^k)$-isogeny between their jacobians.
\end{problem}

\begin{problem}\label{prob2}
Given any superspecial genus-2 curve $C_1$ defined over $\FF_{p^2}$, find
\begin{enumerate}
    \item a curve $C_2$ and a $(2^{k}, 2^{k})$-isogeny $J_{C_1} \rightarrow J_{C_2}$,
    \item a curve $C'_2$ and a $(2^{k'}, 2^{k'})$-isogeny $J_{C_1} \rightarrow J_{C'_2}$,
\end{enumerate}
such that $C_2$ and $C_2'$ are $\overline{\FF}_p$-isomorphic. Here, it is allowed that $k=k'$ but in this case the kernels should be different.
\end{problem}

They are related to our hash function in the following way.

\begin{itemize}
    \item \textit{Preimage resistance:} Finding a preimage in our hash function implies a solution to Problem \ref{prob1} with $C_1=C_0$ as follows. Let $C_2$ be a representative of the isomorphism class of the output of the hash function. A preimage for that output corresponds to a path of length $k$ in our graph, or equivalently, a $(2^k, 2^k)$-isogeny between the jacobians of $C_0$ and $C_2$.
    
    \item \textit{Collision resistance:} Finding a collision in our hash function implies a solution to Problem \ref{prob1} with $C_1=C_0$ as follows. A collision in our hash function corresponds to two distinct paths in our graph with the same ending vertex. Equivalently this amounts to a pair of isogenies 
    \[ \phi: J_{C_0}\rightarrow J_{C_2} \quad \text{and} \quad \phi': J_{C_0}\rightarrow J_{C'_2} \] 
    of type $(2^k, 2^k)$ resp $(2^{k'}, 2^{k'})$ such that $C_2 \cong C_2'$, and with different kernels.
\end{itemize}

To our knowledge, there are no known ways to find isogenies of the said kinds between jacobians of (superspecial) genus-2 curves which perform better than the generic attacks.
In the classical case the best known such attack is
Pollard-rho, which can find a collision or
preimage in time complexity the square root of the number of
possibilities times the amount of time that one step computation takes. In our case we have a graph of size $O(p^3)$ and one step is simply a polynomial computation with some constants, which we can perform in time complexity $\log p$. Hence a Pollard-rho attack could find a solution to Problem \ref{prob1} or Problem \ref{prob2} in time $\tilde O(p^{3/2})$.

With quantum computers in mind, the best known attack is a claw finding algorithm to find a collision or preimage in the graph $\mathcal G_p$ \cite{clawfinding}. An attack like that has time complexity the third root of the size of the graph we work over, instead of the square root in the classical case. This implies we could find a solution to Problem \ref{prob1} or Problem \ref{prob2} in time $\tilde O(p)$.

\section{
    Implementation and timings
}
\label{section:implementation_and_timings}

We have implemented our hash function in Magma, taking into account all
the choices made from the previous section. The pseudocode can be found
below; the Magma code can be found in Appendix \ref{hashfunction}.
The subroutine \texttt{Factorization} is defined as follows:
when the input is a quadratic polynomial,
\texttt{Factorization} returns its two linear factors
(which, in this application, are guaranteed to exist over the ground field).
When the input is a linear polynomial,
it returns that polynomial and \(1\).

\begin{algorithm}
    \caption{Hashing a message $m$ using Richelot isogenies, with $\lambda$ bits of security on a classical computer}
    \label{alg:hash}
    \KwData{Message $m$ and security parameter $\lambda$}
    \KwResult{The hash of $m$ using Richelot isogenies in a graph
    $\thegraph_p$, or \(\bot\) (failure)}
    $S \gets [ (\{1,3\},\{2,5\},\{4,6\}), \ldots, (\{1,6\},\{2,4\},\{3,5\})]$
    \;
    $p \gets$ the smallest prime such that $p>2^{\lceil2\lambda/3\rceil}$ and $p\equiv 5\bmod 6$
    \;
    $(L_1,L_2,L_3,L_4,L_5,L_6) \gets (x-1, x+1, x, x-2, x-1/2, 1) \in \FF_{p^2}[x]^6$
    \;
    $m \leftarrow 2^{30}m$
    \;
    \While{$m > 0$}{
        $i \gets m \bmod 8$
        \;
        $m \gets (m-i)/8$
        \;
        $[G_1, G_2, G_3] \gets $ pairwise products of the \(L_j\) according to $S[i]$
        \;
        \(
            (L_1,L_2)
            \gets
            \texttt{Factorization}(H_1)
        \)
        \textbf{where} 
        \(H_1 := G_2'G_3 - G_2G_3'\)
        \label{alg:hash:sqrt-1}
        \;
        \(
            (L_3,L_4)
            \gets
            \texttt{Factorization}(H_2)
        \)
        \textbf{where} 
        \(H_2 := G_3'G_1 - G_3G_1'\)
        \label{alg:hash:sqrt-2}
        \;
        \If{$\mathrm{rank}(H_1,H_2) = 1$}{
            \Return{\(\bot\)}
            \tcp*{We have hit a vertex in \(\thegraphell_p\)}        
        }
        \(
            (L_5,L_6)
            \gets
            \texttt{Factorization}(H_3)
        \)
        \label{alg:hash:sqrt-3}
        \textbf{where} 
        \(H_3 := G_1'G_2 - G_1G_2'\)
        \;
    }
    \Return{invariants of genus-2 curve defined by the equation $y^2 =
    \prod_{j=1}^6L_j$}
\end{algorithm}


\begin{remark}
    We don't keep track of the leading coefficient of the polynomial determining the genus-2 curve, for the reason that a twist of a curve does not change its Cardona--Quer invariants anyway.
    Similarly we never need to know the exact value of $\delta = \det(G_1, G_2, G_3)$. We are only interested in whether or not $\delta$ equals 0, and with the formulas from the preliminaries, this condition can be easily verified to be equivalent to all $H_i$ being a multiple of one another. Hence it suffices to check if rank$(H_1, H_2) < 2$ instead.
\end{remark}


The deterministic choice of ordering the edges depends on 2 things. First, there's the (arbitrary) way we hardcoded the set $S$, which denotes the pairs of indices of the allowed quadratic splittings. Secondly, the subroutine \texttt{Factorization} automatically orders the roots of the polynomial in some way. In this statement we silently assumed that this happens deterministically by the used software, which is the case for Magma.

Note that we do not claim this code is optimized in any way. For example
we simply pick the smallest prime $p$ possible that satisfies our needs,
whereas better choices may speed up the arithmetic in the field we work
over. Additionally, we did not implement any proper padding schemes, nor did we
make the function constant time (though this is not required for
public inputs). The main goal of the implementation is to see what the order of magnitude is for the speed of the hash function and we leave possible optimizations for future work.

As a final remark we want to point out that the output of the hash
function is dependent on the security level required. The output is a
triple in a quadratic field extension of a finite field of
characteristic roughly $2\lambda/3$ bits in case of classical security.
This means our output has bit length $4\lambda$, even though the 
number of possible hash values is only $2\lambda$ bits. It may be possible to compress this but we leave this discussion for future research, too.

The implementation of our genus two CGL hash function algorithm was done in Magma (version 2.32-2) on an Intel(R) Xeon(R) CPU E5-2630 v2 @ 2.60GHz with 128 GB memory. For every prime size we averaged the speed over 1000 random inputs of 100 bits. A summary of our timed results can be found in the following table.

\begin{center}
\begin{tabular}{|c|c|c|c|c|}
  \hline
  & $p\approx 2^{86}$ & $p\approx 2^{128}$ & $p\approx 2^{171}$ & $p\approx 2^{256}$ \\
  \hline \hline
  bits of classical security & 128 & 192 & 256 & 384 \\
  \hline
  bits of quantum security & 86 & 128 & 170 & 256 \\
  \hline
  time per bit processed & 5.01ms & 6.52ms & 9.33ms & 15.70ms \\
  \hline
  output bits & 516 & 768 & 1026 & 1536 \\
  \hline
\end{tabular}
\end{center}

\section{
    Comparison to Charles--Goren--Lauter, and concluding remarks
}
\label{section:remark_improvement}

    The computational cost of each iteration of the main loop in
    Algorithm~\ref{alg:hash}
    is dominated by the costs of the three square roots required
    to factor the \(H_i\) in Lines~\ref{alg:hash:sqrt-1},
    \ref{alg:hash:sqrt-2}, and~\ref{alg:hash:sqrt-3}.
    At first glance, this would appear to give no advantage over
    the Charles--Goren--Lauter hash function:
    we compute essentially one expensive square root per bit of hash
    input.
    However, there are two important remarks to be made here:
    \begin{enumerate}
        \item
            The entropy in the Charles--Goren--Lauter hash function is linear in $p$, whereas in our case it is cubic in $p$. This implies that for the same security parameters we can work over much smaller finite fields, so the square roots are substantially easier to compute.
        \item
            The square roots, along with the \(H_i\),
            can be computed completely independently.
            The algorithm therefore lends itself well to 
            three-well parallelization,
            as well as to vectorization techniques
            on suitable computer architectures.
    \end{enumerate}
    
    From this point of view, our proposal is a conjecturally secure version of an ill-constructed hash function that we could call 3CGL, where the message $m$ is split up in $3$ chunks $m_1, m_2, m_3$. Each of these $m_i$ is then hashed using Charles, Goren and Lauter's hash function into a supersingular $j$-invariant $j_i$, resulting in a combined hash value $(j_1, j_2, j_3) \in \FF_{p^2}$. Note that, here too, the number of possible outcomes is $O(p^3)$. However, the security of 3CGL clearly reduces to the problem of finding collisions or pre-images for one of the chunks, which Pollard-rho can do in time $\tilde{O}(p^{1/2})$, compared to $\tilde{O}(p^{3/2})$ in our case.

    While this convinces us that genus 2 hash functions deserve their place in the arena of isogeny-based cryptography, more research is needed to have a better assessment of their security and performance. One potentially interesting track is to adapt Doliskani, Pereira and Barreto's recent speed-up to  Charles, Goren and Lauter's hash function from~\cite{DPB}, which has the appearance of an orthogonal improvement that may also apply to genus $2$. From a security point of view, it would be interesting to understand to what extent the discussion from~\cite{KLPT,petit}, transferring the elliptic curve analogs of Problems~\ref{prob1} and~\ref{prob2} to questions about orders in non-commutative algebras and raising some concerns about using special starting curves, carries over to genus $2$.

\subsection*{Acknowledgements}

We are grateful to Yan Bo Ti for sharing with us a preliminary copy of~\cite{flynnti} and to Frederik Vercauteren for helpful feedback with regards to this paper. This work was supported in part by the Research Council KU Leuven grants C14/18/067 and STG/17/019.

\bibliographystyle{plain}
\bibliography{bib}

\appendix
\newpage
\section{
    Proof of Theorem \ref{grobnertheorem}
}
\label{sec:grobnerproof}

We now settle part (1) of Theorem~\ref{grobnertheorem}, as an immediate consequence to:

\begin{theorem}
Let $C$ be a genus-$2$ curve over a field $K$ of characteristic different from $2$ and $5$. Then the number of outgoing $(2,2)$-isogenies with codomain a product of elliptic curves is at most $6$.\label{grobnerapp}
\end{theorem}

\begin{proof}
We can assume that $K$ is algebraically closed, so that $C$ admits a model of the form $y^2 = \prod_{i=1}^6 (x-\alpha_i)$ for roots $\alpha_i \in K$ satisfying
\[ \prod_{1\leq i<j\leq 6}(\alpha_i-\alpha_j) = 1. \]
Due to the formulas for Richelot isogenies, the
number of $(2,2)$-isogenies with codomain a product of elliptic
curves is determined by how many among the 15 different equations of
the form
\begin{equation}
    \det
    \begin{pmatrix}
        1 & \alpha_{\sigma(1)}+\alpha_{\sigma(2)} & \alpha_{\sigma(1)}\alpha_{\sigma(2)}\\
        1 & \alpha_{\sigma(3)}+\alpha_{\sigma(4)} & \alpha_{\sigma(3)}\alpha_{\sigma(4)}\\
        1 & \alpha_{\sigma(5)}+\alpha_{\sigma(6)} & \alpha_{\sigma(5)}\alpha_{\sigma(6)}
    \end{pmatrix}
    = 0\,,
    \label{detsystem}
\end{equation}
where $\sigma$ is a permutation of $\{1,2,3,4,5,6\}$, can be simultaneously satisfied.

To show that no more than $6$ can occur we work with Gr\"obner
bases. The permutations of equation \eqref{detsystem} determine, up
to sign, $15$ different polynomials $f_1,\ldots,f_{15}$ in
$\FF[\alpha_1,\ldots,\alpha_6]$, where $\FF$ is the prime subfield of $K$. We pick a subset of $7$ of these equations and
form the ideal $I \subset \FF[\alpha_1,\ldots,\alpha_6]$ generated by them. We then add the polynomial $\rho = \prod_{i,j} (\alpha_i - \alpha_j) - 1$
as a generator of $I$ as well. 
Now we determine a Gr\"obner basis $G$ for $I$. If $G=\{1\}$ then the
variety defined by $I$ is empty and hence those $7$ equations we
chose can not be satisfied simultaneously, under the assumption that
all $\alpha_i$ are different. If we repeat this process for all
possible subsets of $7$ equations and find $G=\{1\}$ in all cases,
then we are done. There are $\binom{15}{7}=6435$ possible ways of
selecting such a subset, but this is not a problem for
Magma.\footnote{Remark that by using the symmetry in the variables,
it is possible to reduce the number of case distinctions needed, but
we see no need to optimize this since it is a one time computation.} 

When running the algorithm we choose $\FF = \QQ$, for which we indeed find $G = \{1\}$ in each of the cases. This only shows that there are no solutions if $K$ is of characteristic $0$, while we typically want to work over a field with
prime characteristic. If the Gr\"obner basis $G$ equals $\{1\}$
however, we can write $1$ as linear combination of that particular
choice of polynomials $f_i$, say for example $1 =
h_1f_1+\ldots+h_7f_7 + h_8\rho$. If we then multiply both sides of the
equations by the lowest common multiple of the denominators of the
coefficients of the $h_i$, we obtain an equation with coefficients
in $\ZZ[\alpha_1,\ldots,\alpha_6]$. So as long as the characteristic
$p$ of the field we work over does not divide the lowest common
multiple of the denominators of the coefficients of those $h_i$, we
still find a contradictory system. Hence it suffices to keep track of
the primes that divide the denominators. The resulting primes are
$2$, $3$, $5$, $7$ and $11$. It then suffices to rerun the Gr\"obner basis computations for $\FF = \FF_p$ with $p=3,7,11$, leading to the desired conclusion.
\end{proof}

\label{grobnerproof}
The following is the Magma code that was used. The specific cases $p\in\{3,7,11\}$ can be checked by replacing \verb|Rationals()| by \verb|GF(p)| for any one value of $p$, and by removing the innermost loop that starts with \verb|for coord in c do| completely.

\footnotesize
\vspace{10pt}
\begin{spverbatim}
Q<a1,a2,a3,a4,a5,a6> := PolynomialRing(Rationals(),6);
S := {1,2,3,4,5,6};
I := {};

for sub1 in Subsets(S,2) do
  subseq1 := SetToSequence(sub1);
  for sub2 in Subsets(S diff sub1, 2) do
    subseq2 := SetToSequence(sub2);
    subseq3 := SetToSequence(S diff (sub1 join sub2));
    M := Matrix(Q,3,3,
               [ 1, Q.subseq1[1] + Q.subseq1[2], Q.subseq1[1]*Q.subseq1[2], 
                 1, Q.subseq2[1] + Q.subseq2[2], Q.subseq2[1]*Q.subseq2[2], 
                 1, Q.subseq3[1] + Q.subseq3[2], Q.subseq3[1]*Q.subseq3[2] ] );
    eqn := Determinant(M);
    if -eqn notin I then
      I join:= {Determinant(M)};
    end if;
  end for;
end for;

disc := Q ! 1;
for sub in Subsets(S,2) do
  subseq := SetToSequence(sub);
  disc *:= Q.subseq[1] - Q.subseq[2];
end for;

groebnerboolean := true;
badprimes := {};
for j in Subsets(I,7) do
    J := {disc-1};
    J join:= j;
    if GroebnerBasis(Ideal(J)) ne [1] then groebnerboolean := false; end if;
    J := IdealWithFixedBasis(SetToSequence(J));
    c := Coordinates(J, Q ! 1);
    for coord in c do
        for coeff in Coefficients(coord) do
            badprimes join:= SequenceToSet(PrimeDivisors(Denominator(coeff)));
        end for;
    end for;
end for;
print groebnerboolean; badprimes;
\end{spverbatim}
\normalsize
\vspace{10pt}
Theorem~\ref{grobnerapp} cannot be proved in this way for $p=2$,
because equations for hyperelliptic curves are a lot more involved in
fields of even characteristic. Theorem~\ref{grobnertheorem} remains true in this case however, since there are no superspecial genus-2 jacobians in fields of even characteristic.

The following example shows why
Theorem~\ref{grobnertheorem} is not true for $p=5$,
and also provides an example to show that the bound of 6 is sharp.

\begin{example}
    Let $C$ be the genus-2 curve given by $y^2=x^5-x$ over $\FF_{p}$
    (which is superspecial when \(p \equiv 5 \bmod{8}\)),
    and let $i \in \FF_{p^2}$ be a square root of $-1$.
    Of the fifteen quadrating splittings of \(x^5 - x\),
    the six splittings
    \begin{align*}
        \{x,x^2 - (i + 1)x + i,x^2 + (i + 1)x + i\}
        \,,
        &
        \quad
        \{x,x^2 + (i - 1)x - i,x^2 - (i - 1)x - i\}
        \\
        \{x - 1,x^2 + 1,x^2 + x\}
        \,,
        &
        \quad
        \{x + 1,x^2 + 1,x^2 - x\}
        \,,
        \\
        \{x - i,x^2 - 1,x^2 + ix\}
        \,,
        &
        \quad
        \{x + i,x^2 - 1,x^2 - ix\}
    \end{align*}
    all have \(\delta = 0\), so they are always singular.
    The quadratic splitting
    \(\{x,x^2 + 1,x^2 - 1\}\)
    has \(\delta = \pm 2\) (the sign of \(\delta\) may change with the
    order of the factors), and so is never singular.
    There are eight splittings remaining.
    The four splittings
    \begin{align*}
        \{x - 1,x^2 - ix,x^2 + (i + 1)x + i\}
        \,,
        \quad
        \{x - i,x^2 + x,x^2 + (i - 1)x - i\}
        \,,
        \\
        \{x + 1,x^2 + ix,x^2 - (i + 1)x + i\}
        \,,
        \quad
        \{x + i,x^2 - x,x^2 - (i - 1)x - i\}
        \phantom{\,,}
        \intertext{
            all have \(\delta = \pm(3i + 1)\),
            while their ``conjugates'', the four splittings 
        }
        \{x - 1,x^2 + ix,x^2 - (i - 1)x - i\}
        \,,
        \quad
        \{x + i,x^2 + x,x^2 - (i + 1)x + i\}
        \,,
        \\
        \{x + 1,x^2 - ix,x^2 + (i - 1)x - i\}
        \,,
        \quad
        \{x - i,x^2 - x,x^2 + (i + 1)x + i\}
        \phantom{\,,}
    \end{align*}
    have \(\delta = \pm(3i-1)\).

    Now, when \(p = 5\),
    we may take \(i = 2\)
    or \(i = 3\).
    If \(i = 2\) then \(3i-1 = 0\),
    so the last set of four become singular
    (and the penultimate set of four have \(\delta = \pm2\)),
    while if \(i = 3\) then \(3i + 1 = 0\),
    so the penultimate set of four become singular
    (and then the last set of four have \(\delta = \pm2\)).
    In either case, for \(p = 5\) we have exactly four additional singular splittings,
    making ten in total;
    and we cannot have \(i = 2\) or \(3\) in any other characteristic,
    so if \(p \not= 5\) then there are only six singular splittings.
%
\end{example}

\newpage

\section{
    Hash function
}
\label{hashfunction}

The following is the Magma code that implements the hash function we described with the specific choices we have made.

\footnotesize
\vspace{10pt}
\begin{spverbatim}
function G2CGLhash(lambda, message)

splits := [ [{1,3},{2,5},{4,6}], [{1,3},{2,6},{4,5}],
[{1,4},{2,5},{3,6}], [{1,4},{2,6},{3,5}],
[{1,5},{2,3},{4,6}], [{1,5},{2,4},{3,6}],
[{1,6},{2,3},{4,5}], [{1,6},{2,4},{3,5}] ];

p:= 2^Ceiling(lambda*2/3); repeat p := NextPrime(p); until p mod 6 eq 5;
F<a> := GF(p^2);
R<x> := PolynomialRing(F);
factors := [x-1, x+1, x, x-2, x-1/2, 1];
mbase8 := []; message := message*2^30;
while message gt 0 do Append(~mbase8, message mod 8); message div:= 8; end while;

function fac(pol)
r := [ rt[1] : rt in Factorization(pol)];
if #r eq 1 then Append(~r,1); end if;
return r;
end function;

for i := 1 to #mbase8 do
  split := splits[mbase8[i]+1];
  G1 := &*[ factors[j] : j in split[1]];
  G2 := &*[ factors[j] : j in split[2]];
  G3 := &*[ factors[j] : j in split[3]];
  h1 := Derivative(G2)*G3 - G2*Derivative(G3); r1 := fac(h1);
  h2 := Derivative(G3)*G1 - G3*Derivative(G1); r2 := fac(h2);
  if Rank(Matrix(F, 2, 3, [ [Coefficient(h1,j) : j in [0..2]],
                            [Coefficient(h2,j) : j in [0..2]] ])) eq 1 then
    print "No hash for this value possible."; return 0;
  end if;
  h3 := Derivative(G1)*G2 - G1*Derivative(G2); r3 := fac(h3);
  factors := r1 cat r2 cat r3;         
end for;

return G2Invariants(HyperellipticCurve(&*factors));

end function;


\end{spverbatim}
\normalsize

\end{document}